\documentclass[preprint, 3p]{elsarticle} % Really small margins, 2 columns
\usepackage[english]{babel}
\usepackage{amsmath}
\usepackage{amssymb}
\usepackage{amsthm}
\usepackage[dvipsnames]{xcolor}
\usepackage{lineno}
\usepackage[hidelinks]{hyperref}
\usepackage[capitalize]{cleveref}
\usepackage{graphicx}
\usepackage{siunitx}
\usepackage{url}
\graphicspath{{./figures/}}

\newtheorem{proposition}{Proposition}

\journal{Automatica}
\date{October 7, 2024}

\usepackage[normalem]{ulem} % Striking out text
\begin{document}
%% Frontmatter
\begin{frontmatter}

\title{Non-collocated vibration absorption using delayed resonator for spectral and spacial tuning -- analysis and experimental validation}

% \affiliation[label1]{organization={},
\author[CTU]{Matěj Kuře}
\author[CTU]{Adam Peichl}
\author[CTU]{Jaroslav Bušek}
\author[UCO]{Nejat Olgac}
\author[CTU,CIIRC]{Tomáš Vyhlídal}
\affiliation[CTU]{
    organization={Department of Instrumentation and Control Engineering, Faculty of Mechanical Engineering, Czech Technical University in Prague},
    %addressline={Technická 1902/4}, 
    city={Prague~6},
    postcode={160~00},
    country={Czech Republic},
}
\affiliation[CIIRC]{
    organization={Czech Institute of Informatics, Robotics and Cybernetics, Czech Technical University in Prague},
    %addressline={Jugoslávských partyzánů 1580/3}, 
    city={Prague~6},
    postcode={160~00}, 
    country={Czech Republic}
}
\affiliation[UCO]{
    organization={Department of Mechanical Engineering, University of Connecticut},
    %addressline={\textcolor{red}{Address Two}},
    city={Storrs},
    %postcode={22222}, 
    state={CT},
    country={USA}
}

%% Abstract
\begin{abstract}
Non-collocated vibration absorption (NCVA) concept using delayed resonator for in-situ tuning is analyzed and experimentally validated. There are two critical contributions of this work. One is on the scalable analytical pathway for verifying the concept of \emph{resonant substructure} as the basis of the ideal vibration absorption. The second is to experimentally validate the \emph{spatial and spectral} tunability of NCVA structures for the first time. For both novelties arbitrarily large dimensions of interconnected mass-spring-damper chains are considered. Following the state of the art on NCVA, control synthesis is performed over the \emph{resonant substructure} comprising the delayed resonator and a part of the primary structure involved in the vibration absorption. The experimental validation of the proposed NCVA concept is performed on a mechatronic setup with three interconnected cart-bodies. Based on the spectral analysis, an excitation frequency is selected for which a stable vibration suppression can be achieved sequentially for all the three bodies, one collocated and two non-collocated. The experimental results closely match the simulations for complete vibration suppression at the targeted bodies, and thus validating the crucial \emph{spatial} tunability characteristic as well as the traditional \emph{spectral} tuning.
\end{abstract}

%% Keywords

\tnotetext[footnoteinfo]{%
The presented research was supported by the Czech Science Foundation under the project 21-00871S and by the European Union under the project ”Robotics and advanced industrial production”, reg. no. CZ.02.01.01/00/22\_008/0004590. The first and third authors also acknowledge support by the Grant Agency of the Czech Technical University in Prague, student grant No.~SGS23/157/OHK2/3T/12.\\ \\
\textbf{Copyright:}\textcolor{white}{1}\copyright\textcolor{white}{1}2024 by the authors. This manuscript version is made available under the CC-BY 4.0 license \href{https://creativecommons.org/licenses/by/4.0/}{https://creativecommons.org/licenses/by/4.0/}}
\end{frontmatter}

% % To be deleted -- layout information
% \noindent
% \colorbox{yellow}{%
% LW: \the\linewidth, 
% TW: \the\textwidth, 
% TH: \the\textheight
% }

%% Introduction ================================================================
\section{Introduction}
\label{sec:Introduction}

Vibration absorbers have proven to be effective in a variety of engineering applications \cite{preumont_2018_vibration}. In the traditional \emph{collocated vibration absorption} task, the absorber is deployed at the place of the mechanical structure, where vibration is to be suppressed. The spectral tuning of the collocated absorber is a relatively straightforward task and has been widely addressed in literature, see e.g., \cite{rana1998parametric}, \cite{richiedeiTunedMassDamper2022} with passive, \cite{yuan2019mode} with semi-active, and \cite{preumontVibrationControlActive2011} with active tuning methods. In many applications, however, due to operational reasons, the absorber needs to be deployed in a \emph{non-collocated} manner, i.e., at a different location from the vibration suppression target. The design and the tuning of such an absorber is a considerably more difficult task because a part of the primary structure between the absorber and the suppression target has to be engaged in the action.

In this paper, we focus on the analysis and the experimental validation of \emph{non-collocated vibration absorption} (NCVA) using the \emph{delayed resonator} (DR) tuning procedure, which was presented recently by Olgac and Jenkins, \cite{jenkins2019real}, \cite{olgacActivelyTunedNoncollocated2021}, see also \cite{olgac2022benchmark}. The DR tuning scheme was proposed in the 90s by Olgac and his co-workers primarily for \emph{collocated vibration absorption}. Since then, it has become a traditional tool for (collocated) vibration absorption and a benchmark case revealing potential benefits of involving time delays in the feedback control law. DR is an active vibration absorber which is created using a decentralized time-delayed feedback scheme executed on the absorber's position \cite{olgac1994novel}, velocity \cite{filipovic2002delayed}, or acceleration \cite{olgac1997active2}. The time-delayed feedback is applied to turn the absorber substructure into an ideal resonator which completely suppresses the vibration. Please note that, traditionally, the DR is always tuned to be marginally stable (i.e., with a characteristic root pair placed at $\pm \jmath \omega$, $\omega$ being the excitation frequency) \cite{inman2022engineering}. This concept is extended to NCVA deployment, this time however, entailing the absorber as well as a part of the primary structure. This newly composed segment was named as the \emph{resonant substructure} \cite{olgacActivelyTunedNoncollocated2021}. In this paper, we present another analytical pathway reinforcing this critical aspect.

A practical benefit of the DR is that in the standard collocated setting, neither measurements at the primary structure, nor its physical parameters are involved in the DR design. From the wide literature on the DR, let us mention a torsional absorber \cite{hosek1997tunable}, an auto-tuning algorithm to enhance the robustness against uncertainties \cite{hosek2002single} and multiple DR application  \cite{jalili1999multiple}. Recent DR design and analysis topics include stability analysis \cite{vyhlidal2019analysis}, combination of position and velocity feedback \cite{oytun2018new}, DR with distributed delays \cite{7378505}, \cite{liu2024delayed}, targeting two frequencies \cite{valasek2019real}, enhancing the robustness in vibration absorption \cite{PilbauerRobust}, \cite{kure2024robust}, fractional order DR \cite{cai2023spectrum}, and the DR concept extension to two \cite{vyhlidalAnalysisOptimizedDesign2022}, \cite{vsika2021two}, and three \cite{vsika2024three}, \cite{benevs2024collocated} dimensional vibration absorption. Let us also point to delay-free resonator alternatives to the DR. In \cite{rivaz2007active} a PI acceleration feedback of the absorber was proposed, supplemented by a low and high pass filters. In \cite{filipovic1999vibration}, the concept of \emph{linear active resonator} (LAR) was introduced by Filipovic and Schröder. Conceptually, it mirrors the DR structure with a tuneable gain, which, however, is in a series with rational transfer function instead of the sole delayed term used in DR. %Thus, its implementation is not as straightforward as that of the DR.  

After outlining the development and recent topics on collocated DR concept, let us turn the attention to the non-collocated case. In \cite{jenkins2019real}, \cite{olgacActivelyTunedNoncollocated2021}, Olgac and Jenkins demonstrated that the DR is applicable for non-collocated vibration absorption of a system composed of a serial interconnection of flexibly linked masses. They showed that compared to the collocated DR design, the part of the primary structure between the position to be silenced and the position where the DR is deployed needs to be included in tuning the control logic. The DR together with this part of the primary forms the \emph{resonant substructure}, which needs to be tuned as a whole. It was also discussed in \cite{jenkins2019real}, \cite{olgacActivelyTunedNoncollocated2021} that the resonant substructure can only be identified under some restrictions on the physical deployment. Note that the findings of Olgac and Jenkins confirm the earlier results by Filipovic and Schröder presented in \cite{filipovic2001control}, where an analogous problem of remote (non-collocated) vibration suppression at a system composed of series of flexibly linked masses is solved by the LAR. 
% Next to the theoretical analysis performed over the transfer functions, an experimental validation is performed for the two-mass setup. 

Subsequent to works by Olgac and Jenkins, in \cite{silm_2024_spectral_design}, a spectral design of non-collocated vibration suppression performed primarily by a DR is presented. The method is based on a purely imaginary pair of active zero assignment to the transfer function between the excitation force and the target position to be silenced. As such, it is also applicable to setups, where the resonant substructure cannot be defined. 
For such cases, in order to increase the stability margin, an additional stabilizing controller is included and tuned. In the follow-up work \cite{Saldanha2023}, an output feedback controller is used to assign the active zero couple and to stabilize the system with feedback delay. The synthesis is performed by a spectral optimization. The results of both \cite{silm_2024_spectral_design} and \cite{Saldanha2023} are experimentally validated on a mechatronic setup with cart-bodies targeting vibration suppression at a single body. In \cite{saldanha2023IFAC} the simultaneous imaginary zero assignment and stabilization at the non-collocated vibration suppression is achieved by a DR with multiple static delay feedback. 

In this paper, building up on the results by Olgac and Jenkins, \cite{jenkins2019real}, \cite{olgacActivelyTunedNoncollocated2021}, we analyze further the problem of non-collocated vibration absorption utilizing delayed resonator. Instead of acceleration delayed feedback considered in Olgac and Jenkins, we consider position delayed feedback to tune the absorption properties of the resonant substructure. Compared to \cite{jenkins2019real}, \cite{olgacActivelyTunedNoncollocated2021}, and also to \cite{filipovic2001control}, in the problem analysis and control synthesis, we avoid derivation of the transfer function from the continuous time model. The analysis and tuning is performed directly over the system model matrices which makes it numerically efficient even for large number of masses. An eminent novelty of this paper stems in the experimental validation of the non-collocated vibration absorption by the DR. It is performed on a mechanical system composed of a series of cart-bodies connected with springs elements and actuated by voice-coils. It is shown that for a selected excitation frequency, almost ideal vibration suppression of any of the carts, both collocated and non-collocated, vis-a-vis the DR position, can be achieved.

The rest of the paper is composed as follows. In \cref{sec:main_result:tuning_DR}, the problem of targeted NCVA is formulated. Subsequent \cref{sec:control-law-desing} demonstrates the key role of resonant substructure and outlines the design of the delayed position feedback. A thorough experimental validation is presented in \cref{sec:case_study}, and in the last \cref{sec:conclusion}, a summary and further research directions are provided.

%% MAIN RESULT =================================================================
\section{Problem formulation}
\label{sec:main_result:tuning_DR}
\begin{figure*}
    \centering
    \includegraphics[scale=1]{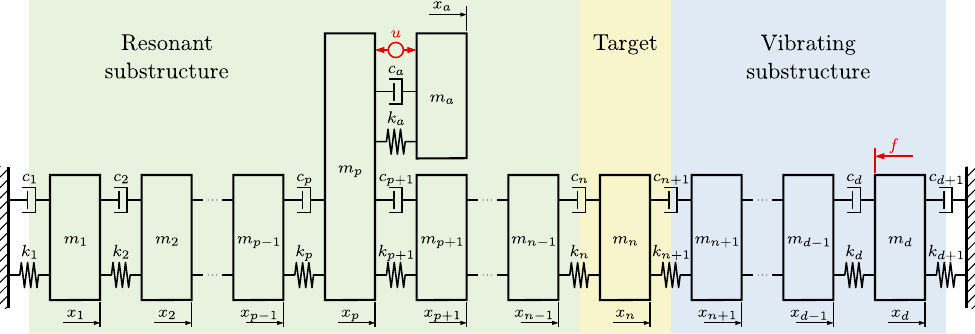}% Original size 468x160pt
    \caption{A general scheme of multi-body structure consisting of a series of linearly linked masses $m_1,\ \ldots,\ m_d$, being excited by an $\omega$-harmonic force $f$, together with an active absorber $m_a$. The structure can be split into: i) resonant substructure tuned by inner feedback $u(t)$ to resonate at frequency $\omega$, ii) target mass $m_n$ to be silenced, which is non-collocated with the absorber deployment at $m_p$, and iii) vibrating substructure.}
    \label{fig:splitting:general_scheme_split}
\end{figure*}

Consider a mechanical structure shown in \cref{fig:splitting:general_scheme_split} composed of $d$ linearly chained masses $m_i$ which are interconnected to its neighbors via springs $k_i$, $k_{i+1}$ and dampers $c_i$ and $c_{i+1}$. The first and the last masses $m_1$ and $m_d$ are connected to a rigid frame. A harmonic disturbance force
\begin{equation}
    \label{eq:general:excitation}
    f(t) = F \cos \left(\omega t \right)
\end{equation}
with amplitude $F$ and frequency $\omega$ acts on the last mass $m_d$ causing the whole structure to vibrate. A DR absorber with mass $m_a$ is deployed at the mass $m_p$ through the spring $k_a$, the damper $c_a$ and the actuator $u(t)$. The target NCVA mass is denoted as $m_n$. 

The setup is modeled by a set of second-order linear equations
\begin{equation}\label{eq:systemmodel}
    M\ddot{x}(t) + C\dot{x}(t) + Kx(t) = B_f f(t) + B_u u(t),
\end{equation}
where $x(t)=\begin{bmatrix}x_a(t) & x_1(t) & \dots & x_d(t) \end{bmatrix}^\mathsf{T}$ is a vector of displacements of the masses, mass matrix is given by
\begin{equation}
     M = \text{diag}\left(m_a, m_1,  \dots, m_p, \dots, m_n, \dots, m_{d-1}, m_d\right),
  \end{equation}
and stiffness $K$ and damping $C$ matrices are composed considering the rules of parallel spring interconnections. The matrices are omitted here due to space constraints, but they are given in the Case study validation section for the considered three-cart configuration.
%\ddeltext{by}
 % \begin{equation}
 %     K
 %     =
 %     \begin{bmatrix}
 %         k_a & & & & -k_a & & & & 0\\
 %         & k_1 + k_2 & -k_2 & & & & & & & \\
 %         & -k_2      & k_2 + k_3 & -k_3 & & & \\
 %         & & \ddots & \ddots & \ddots \\ 
 %         - k_a & & & -k_{p} & k_{p} + k_{p+1} + k_a & -k_{p+1} & & \\
 %         & & & & \ddots & \ddots & \ddots \\
 %         & & & & & -k_{t} & k_{t} + k_{t+1} & k_{t+1}\\
 %         & & & & & & \ddots & \ddots & \ddots \\
 %         & & & & & & -k_{d-1} & k_{d-1} & - k_{d} \\
 %         0 & & & & & & & -k_{d} & k_{d} + k_{d+1} \\
 %     \end{bmatrix},
 % \end{equation}
 %  \colorbox{yellow}{Needs reformulating} \mmodtext{K is a tridiagonal matrix of corresponding couplings that is easy to obtain}
 % while damping matrix $C$ is of a same structure as $K$, but each corresponding stiffness $k_i$ is replaced by $c_i$. 
The input matrices are given as
\begin{align}
     B_f &=E_d, \\
     B_u &=E_a - E_p,
\end{align}
where the vectors encoding position of absorber mass $m_a$ and chain masses $m_i$ are defined as
\begin{align}
     E^\mathsf{T}_a &=\begin{bmatrix} 1 & o_{d}^\mathsf{T} \end{bmatrix}, \\
     E^\mathsf{T}_i &=\begin{bmatrix} o_{i}^\mathsf{T} & 1 & o_{d-i}^\mathsf{T} \end{bmatrix} \ i=1, \dots, d,
\end{align}
with $o_l$ denoting $l$-dimensional zero vector.
 
\section{Control law design}
\label{sec:control-law-desing}
The control objective is to achieve complete vibration suppression at the target mass $m_n$ by applying the DR position feedback
\begin{equation}
    \label{eq:general:cDR_position}
    u(t) = g x_a \left( t - \tau \right)
\end{equation}
with the gain $g$ and the delay $\tau$ being parameters to be tuned. Note that the feedback is taken from the absorber position only and no measurements from the primary structure are considered. Note also that the feedback \eqref{eq:general:cDR_position} differs from that in \cite{jenkins2019real}, \cite{olgacActivelyTunedNoncollocated2021}, where acceleration feedback was used. 

Using $x_a(t) = E^T_a x(t)$ allows us to write characteristic matrix of closed loop
 \begin{equation}\label{eq:setup-char-matrix}
     R(s; g, \tau) = Ms^2 + Cs + K - g B_u E_a^\mathsf{T} e^{-s\tau},
 \end{equation}
 which can be decomposed into resonant-target-vibrating substructures (see \cref{fig:splitting:general_scheme_split}) signing them with R, T and V designations
 \begin{equation}\label{eq:expR}
     R(s; g, \tau) = 
     \begin{bmatrix}
         A_\mathrm{R}(s; g, \tau) & a_\mathrm{R}(s) & O \\
         a_\mathrm{R}^\mathsf{T}(s) & a_\mathrm{T}(s) & a_\mathrm{V}^\mathsf{T}(s) \\
         O^\mathsf{T} & a_\mathrm{V}(s) & A_\mathrm{V}(s)
     \end{bmatrix},
 \end{equation}
 where the dimensions of the functional submatrices are: $n\times n$ for $A_\mathrm{R}(s; g, \tau)$, $n \times 1$ for $a_\mathrm{R}(s)$, $d-n \times 1$ for $a_\mathrm{V}(s)$, $d-n \times d-n$ for $A_\mathrm{V}(s)$, $n \times d-n$ for zero-matrix $O$, and $a_\mathrm{T}(s)$ is scalar.%Note that $n$ is the number associated with the target mass.
 % where all blocks are obtained by taking according slices of $R(s; g, \tau)$, i.e, $A_R(s; g, \tau)$ is obtained by taking first $t$ rows and $t$ columns, $a_R(s)$ is obtained by taking $t$ rows of $(t+1)$th column, $a_t(s)$ is element from $(t+1)$th row and $(t+1)$th column, $A_V(s)$ is obtained by taking last $d-t$ rows and $d-t$ columns and $a_V(s)$ is obtained by taking $d-t$ rows of $(t+1)$th column.
 
The transfer function between the excitation force $f(s)$ and the position of the target mass $x_n(s)$ is given by
\begin{equation}\label{eq:TF}
    P(s; g, \tau)=\frac{x_n(s)}{f(s)}= E_\mathrm{T}^\mathsf{T} R^{-1}(s; g, \tau) B_f.
\end{equation}
% finding suitable DR parameters that fullfill $P(\jmath \omega, g, \tau)=0$. 
From the transfer function analysis, it was shown by both Olgac and Jenkins \cite{jenkins2019real}, \cite{olgacActivelyTunedNoncollocated2021} considering the DR and by Filipovic and Schröder \cite{filipovic2001control}, considering LAR, that the poles of the resonant substructure become the zeros of the transfer function \eqref{eq:TF}. In the following proposition, we confirm the result without the need of deriving the respective transfer function. As such, the validity of this claim can be easily extended towards setups with large number of masses, for which derivation of the transfer function would be cumbersome or even numerically risky.
% However, the control law parameters $g, \tau$ can be determined from resonant substructure part $A_R(s; g, \tau)$ only as will be demonstrated in the following proposition:
\begin{proposition}
    \label{prop:poles_resonating_to_zeros_overall}
    The poles of the resonant substructure, i.e., the roots of the equation
    \begin{equation}\label{eq:prop}
        \det \left( A_\mathrm{R}(s; g, \tau) \right) = 0,
    \end{equation}
    are zeros of the transfer function \eqref{eq:TF}.
\end{proposition}

\begin{proof}
   Let us define
       \begin{equation}\label{eq:prop:z}
        z(s; g, \tau) = \det \left( \begin{bmatrix}
            R(s; g, \tau) & -B_f\\
            E_\mathrm{T} & 0
        \end{bmatrix} \right),
    \end{equation}
    so that the zeros of \eqref{eq:TF} are the roots of the equation
    \begin{equation}\label{eq:prop:zeros}
        z(s; g, \tau)  = 0.
    \end{equation}
    Assuming invertibility of $A_\mathrm{R}(s; g, \tau)$, $a_\mathrm{T}(s)$ and $A_\mathrm{V}(s)$, \eqref{eq:prop:z} can be rewritten to
    \begin{equation}
        \label{eq:prop:z-schur-1}
        z(s) = \det\left( R(s; g, \tau) \right) \left(E_\mathrm{T}^\mathsf{T} R^{-1}(s; g, \tau) B_f \right).
    \end{equation}
    % \begin{equation}
    %     \label{eq:prop:z-schur-1}
    %     z(s) = \det\left( R(s; g, \tau) \right) \det\left( 0 + E^T_t R^{-1}(s; g, \tau) B_f \right).
    % \end{equation}
    % \ddeltext{Both its parts can further be expressed as}
    The first factor of \eqref{eq:prop:z-schur-1} can be expressed using Schur complement as
    \begin{equation}
        \label{eq:prop:z-schur-2}
        \det\left( R(s; g, \tau) \right) = \det \left( A_\mathrm{R}(s; g, \tau) \right) \det \left( H \right),
    \end{equation}
    where
    \begin{equation*}
        H =
        \begin{bmatrix}
            a_\mathrm{T}(s) - a_\mathrm{R}^\mathsf{T}(s) A_\mathrm{R}^{-1}(s; g, \tau) a_\mathrm{R} (s) & a_\mathrm{V}^\mathsf{T}(s) \\
            a_\mathrm{V}(s) & A_\mathrm{V}(s)
        \end{bmatrix}.
    \end{equation*}
    The second factor of \eqref{eq:prop:z-schur-1} can be rewritten into %using the Scur complement again, becomes % schur inverse matrix lemma
    \begin{equation}
        \label{eq:prop:z-schur-3}
        E_\mathrm{T}^\mathsf{T} R^{-1}(s; g, \tau) B_f = \begin{bmatrix}
             1 & o_{d-n}^\mathsf{T}
        \end{bmatrix} H^{-1} \begin{bmatrix}
            0 \\ b_f
        \end{bmatrix},
    \end{equation}
    where
    \begin{equation*}
        E_\mathrm{T} = \begin{bmatrix}
            o_{n} \\ 1 \\ o_{d-n}
        \end{bmatrix},\ 
        B_f = \begin{bmatrix}
            o_{n} \\ 0 \\ b_f
        \end{bmatrix}.
    \end{equation*}
    Substituting \eqref{eq:prop:z-schur-2} and \eqref{eq:prop:z-schur-3} into \eqref{eq:prop:z-schur-1} gives % using the Schur complement backwards
    \begin{align}
        z(s) &= \det \left( A_\mathrm{R}(s; g, \tau) \right) \cdot\nonumber\\
        &\det \left(\begin{bmatrix}
            a_\mathrm{T}(s) - a_\mathrm{R}^\mathsf{T}(s) A_\mathrm{R}^{-1}(s; g, \tau) a_\mathrm{R} (s) & a_\mathrm{V}^\mathsf{T}(s) & 0 \\
            a_\mathrm{V}(s) & A_\mathrm{V}(s) & b_f \\
            1 & o_{d-n}^\mathsf{T} & 0
        \end{bmatrix}\right).
    \end{align}
    Regarding the second determinant, by applying the cofactor expansion along the last row the dependency on $A_\mathrm{R}(s; g, \tau)$ disappears. Thus, the roots of \eqref{eq:prop} are among the roots of \eqref{eq:prop:zeros}, and they form zeros of \eqref{eq:TF}. 
\end{proof}
 
By the \cref{prop:poles_resonating_to_zeros_overall}, the task of complete vibration suppression at $m_n(t)$ reduces to assigning a pair of complex conjugate poles at $\pm \jmath \omega$ to the resonant substructure, i.e. ensuring that $\det \left( A_\mathrm{R}(\jmath\omega; g, \tau) \right) = 0$, which can be rewritten into
\begin{equation}\label{eq:char_matrix-resonating}
    \det \left( M_\mathrm{R}s^2 + C_\mathrm{R}s + K_\mathrm{R} - g b_u e_a^\mathsf{T} e^{-s\tau}\right) = 0,
\end{equation}
where matrices $M_\mathrm{R},\ C_\mathrm{R},\ K_\mathrm{R}$ are obtained by taking first $T$ rows and first $T$ columns from theirs respective counterparts $M,\ C,\ K$, and vectors $b_u$, $e_a$ are obtained taking first $T$ rows from vectors $B_u$ and $E_a$, respectively.

Equation \eqref{eq:char_matrix-resonating} can be rewritten into product of two determinants
\begin{multline*}
    \det \left( M_\mathrm{R}s^2 + C_\mathrm{R}s + K_\mathrm{R} \right) \cdot \nonumber\\ \det \left(I - g \left( M_\mathrm{R}s^2 + C_\mathrm{R}s + K_\mathrm{R} \right)^{-1} b_u e_a^\mathsf{T} e^{-s\tau}\right) = 0,
\end{multline*}
where the first part is independent of DR parameters and therefore can be omitted. Using the Weinstein-Aronszajn identity, the second part can be rewritten into
\begin{equation}\label{eq:weinstein-aroszajn}
    1 - g e_a^\mathsf{T} \left( M_\mathrm{R}s^2 + C_\mathrm{R}s + K_\mathrm{R} \right)^{-1} b_u e^{-s\tau} = 0.
\end{equation}
Rearranging \eqref{eq:weinstein-aroszajn} and imposing a resonant root at $s=\jmath\omega$, we can write
\begin{equation} \label{eq:control_design:g_tau_jomega}
    g \mathrm{e}^{-\jmath \omega \tau} =  p(\jmath \omega),
\end{equation}
where 
\begin{equation*}
    p(\jmath \omega)= e_a^\mathsf{T} \left( M_\mathrm{R}(\jmath\omega)^2 + C_\mathrm{R} (\jmath\omega) + K_\mathrm{R} \right)^{-1} b_f.
\end{equation*}

Solving \eqref{eq:control_design:g_tau_jomega} yields two infinite sets of solutions due to the periodicity of argument and symmetry of modulo, one with a positive gain
\begin{equation}
    \label{eq:control_design:g_1_tau_1}
    g = \left| p(\jmath \omega) \right|,\ \tau = \frac{1}{\omega} \left( -\arg(p(\jmath \omega)) + 2k \pi \right),\\ k \in Z,
\end{equation}
and the other with negative gain
\begin{equation}
    \label{eq:control_design:g_2_tau_2}
    g = -\left| p(\jmath \omega) \right|,\ \tau = \frac{1}{\omega} \left( \pi -\arg(p(\jmath \omega)) + 2k \pi \right),\ k \in Z.
\end{equation}
From the infinite set of delay values, it is advisable to select the smallest possible delay and the corresponding gain. In some frequency ranges, however, the larger delay variants can give better results (demonstrated in the Case study validation section below).

Mounting DR absorber on mass $m_p$ with active feedback \eqref{eq:general:cDR_position} and parameters $g, \tau$ selected as defined above ensures that a pair of conjugate zeros is assigned to the imaginary axis at $\pm \jmath \omega$. Thus, assuming the overall system is stable, this ensures that signal of frequency $\omega$ does not pass through the system and the harmonic response of the target mass is fully suppressed. Clearly this DR synthesis needs to be supplemented with a stability check. Also note that, as per \cite{vyhlidal_2019_analysis}, marginal stability of the DR is preferred even for the collocated vibration suppression. This requirement naturally applies for the non-collocated vibration absorption, where one can expect even stronger dependency of the overall system stability on the stability posture of the resonant substructure.
 
Defining the \emph{spectral abscissa} of the overall system as
\begin{equation}\label{eq:alphaOS}
     \alpha_\mathrm{OS}(g, \tau) = \max \left\{ \Re(s); \det\left(R(s; g, \tau)\right) = 0, \right\}
\end{equation}
where $R(s; g, \tau)$, is given by \eqref{eq:setup-char-matrix}, and the \emph{spectral abscissa} of the resonant substructure
\begin{equation}\label{eq:alphaRS}
    \alpha_\mathrm{RS}(g, \tau) = \max \left\{ \Re(s); \det\left(A_\mathrm{R}(s; g, \tau)\right) = 0, \right\}
\end{equation}
where $A_\mathrm{R}(s; g, \tau)$ is defined in \eqref{eq:expR}, the stability condition of the overall system reads as
\begin{equation}\label{eq:alphaOSstab}
    \alpha_\mathrm{OS}(g, \tau) < 0,
\end{equation}
while the marginal stability condition for the resonant substructure reads as 
\begin{equation}\label{eq:alphaRSstab}
    \alpha_\mathrm{RS}(g, \tau) = 0.
\end{equation}
 
%% CASE STUDY ====================================================================
\begin{figure*}[t]
    \centering
    % Original size 180x84 mm
    \includegraphics[scale=1]{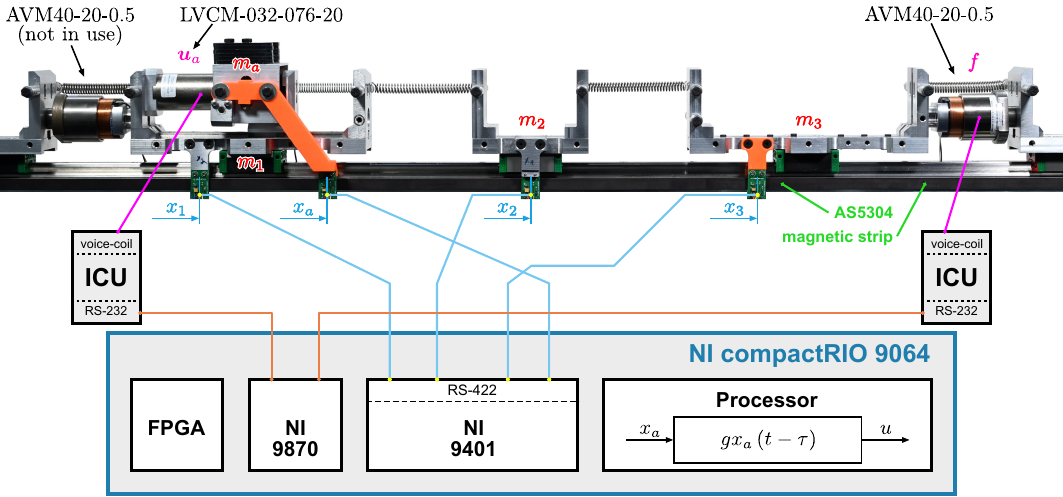}
    \caption{Mechatronic setup designed for experimental validation of non-collocated vibration absorption accompanied with the control implementation scheme.}
    \label{fig:case_study:implementation}
\end{figure*}
\section{Case study validation}
\label{sec:case_study}
The validation of the non-collocated vibration absorption using DR with position feedback is performed on an experimental setup as shown in \cref{fig:case_study:implementation}. With a reference to \cref{sec:main_result:tuning_DR} it consists of three masses, i.e., $d=3$, where the DR is deployed on the first mass, i.e., $p=1$.
%which consists of three carts with the masses $m_i$, $i=1..3$ interconnected by springs and dampers $k_j, c_j, j=1..4$ either to its neighbors or to a fixed frame. An absorber $m_a$ is attached to the cart $m_1$ by means of a spring $k_a$, a damper $c_a$, and an actuator $u(t)$ that is used for active feedback control. 
%Regarding the outputs of this structure, the displacements of all carts $x_1(t)$, $x_2(t)$, $x_3(t)$ and $x_a(t)$ are measured. Last but not least, a harmonic disturbance $f(t)$ is applied to the cart $m_3$. This setup can perform a collocated vibration suppression as well as non-collocated with the offset of 1 or 2 carts.
% \begin{figure}
%    \centering
%    \includegraphics[width=\textwidth]{figures/NCss_Olgac_schema.pdf}
%    \caption{Scheme of the model of the experimental structure.}
%    \label{fig:case_study:validation_scheme}
%\end{figure}
% Mechanical assembly itself
%An experimental testbed has been built following \cref{fig:case_study:validation_scheme} and is depicted in \cref{fig:case_study:implementation}.
The setup consists of a rail to which carts $m_1$, $m_2$ and $m_3$ are attached by industrial ball bearings. Two additional carts with mechanical brake, % instead of ball bearing
one at each end, are attached representing a rigid frame. The absorber $m_a$ is mounted directly on a cart $m_1$ where another (smaller) rail is installed. All carts are flexibly interconnected by springs. 
To measure the displacement of the carts, a multi-pole magnetic strip with a resolution of $\SI{25}{\micro\metre}$ is installed on the setup frame. Each cart is equipped with an AMS AS5304 incremental position sensor with Hall elements reading a quadrature signal. % From Haik's article
% Actuation
Three linear voice-coil motors (LVCM) are installed to actuate the setup. From left to right in \cref{fig:case_study:implementation} we have: 1) Akribys AVM40-20-0.5 LVCM installed between the frame and the cart $m_1$, which is not used in this experiment; 2) 
%Let us note here however, that this motor has % already
%been preinstalled and might be used in the future research to provide a stabilization or to optimize the spectra of the system without getting the mechanical setup disassembled and assembled again. 
Moticont LVCM-032-076-20 used as the DR actuator creating the input $u(t)$; and 3) another Akribys AVM40-20-0.5 generating the disturbance harmonic force $f(t)$. 
These voice-coil actuators are accompanied by two custom made Instrument Control Units ({ICU}s) from PearControl.
% Control and data acquisition
The control algorithms and instrumentation are implemented in LabVIEW\textsuperscript{\texttrademark} 2021 and are executed on the NI compactRIO 9064 industrial control system from National Instruments with a sampling rate of \SI{1}{kHz}. Fast sensor measurement and quadrature signal encoding is executed on an embedded FPGA module with a sampling rate of \SI{48}{MHz}. Two plug-in modules are used in the control unit: NI~9870 used to communicate with the {ICU}s via RS-232 serial lines, and NI~9401 - a digital I/O module used for reading sensors.
% Mechanical notes.
The pull springs are pre-loaded such that the steady-state displacements of the voice coils are in the middle of their strokes. This setting serves for better actuator linearity.

Corresponding to the model in \eqref{eq:systemmodel}, the mass, damping and stiffness matrices are defined as 
\begin{align*}
M & = \text{diag} \left( m_a,\ m_1,\ m_2,\ m_3 \right),\\
     C &
     =
     \begin{bmatrix}
         c_a & -c_a & 0 & 0 \\
         -c_a & c_1 + c_2 + c_a & -c_2 \\
         0 & -c_2 & c_2 + c_3 & -c_3 \\
         0 & 0 & -c_3 & c_3 + c_4
     \end{bmatrix},\\
     K &
     =
     \begin{bmatrix}
         k_a & -k_a & 0 & 0 \\
         -k_a & k_1 + k_2 + k_a & -k_2 \\
         0 & -k_2 & k_2 + k_3 & -k_3 \\
         0 & 0 & -k_3 & k_3 + k_4
     \end{bmatrix},
\end{align*}
with the position vector 
\begin{equation*}
    x(t) = \begin{bmatrix}x_a(t) & x_1(t) & x_2(t) & x_3(t) \end{bmatrix}^\mathsf{T},
\end{equation*}
and input vectors
\begin{equation*}
    B_f=\begin{bmatrix} 0 & 0 & 0 & 1 \end{bmatrix}^\mathsf{T},
\end{equation*}
\begin{equation*}
    B_u=\begin{bmatrix} 1 & -1 & 0 & 0 \end{bmatrix}^\mathsf{T}.
\end{equation*}
In the case study, we will sequentially consider all the three  masses to be stopped: i.e. from collocated vibration absorption ($n = 1$) to non-collocated vibration absorption ($n = 2, 3$). Thus, the system output $y_n(t) = x_n(t)$ varies with respect to the cart to be stopped 
\begin{equation}
    y_n(t) = E_i x(t), \quad i=1,2,3,
\end{equation}
where 
\begin{align*}
    E_1 &= \begin{bmatrix} 0 & 1 & 0 & 0 \end{bmatrix}^\mathsf{T}, \\
    E_2 &= \begin{bmatrix} 0 & 0 & 1 & 0 \end{bmatrix}^\mathsf{T}, \\
    E_3 &= \begin{bmatrix} 0 & 0 & 0 & 1 \end{bmatrix}^\mathsf{T}, \\
\end{align*}
represent the position of the target mass.
% Identification
% \subsection{Identification and parameters}
The structural parameters of the setup are given in \cref{tab:case_study:parameters}. Note that stiffness and damping characteristics are determined experimentally. Clearly, the linear model does not capture all the mechanical phenomena, such as the dry friction of the bearings and non-linearities of the LVCM at large amplitude motion. 
%o obtain the parameters of the experimental structure, each cart was measured to obtain its mass using a digital kitchen scale. Next, step responses of each cart separately and a combination of carts were used to obtain the stiffness and damping of the structure.
%The parameters of the structure are listed in %%\cref{tab:case_study:parameters}.
\begin{table}
    \centering
    \begin{tabular}{|c|c|c|c|}
        \hline
        & mass & stiffness & damping \\
        $i$ & $m_i$ & $k_i$ & $c_i$ \\
        $\left[ \mathrm{-} \right]$ & $\left[ \mathrm{kg} \right]$ & $\left[ \mathrm{N\,m^{-1}} \right]$ & $\left[ \mathrm{N\; s\; m^{-1}} \right]$ \\
        \hline
         a & 0.520 &  407 & 1.80 \\
         1 & 1.175 & 1001 & 4.35 \\
         2 & 0.509 &  749 & 0.85 \\
         3 & 0.705 &  711 & 1.85 \\
         4 &     - &  950 & 4.95 \\
         \hline
    \end{tabular}
    \caption{Identified parameters of the experimental setup}
    \label{tab:case_study:parameters}
\end{table}

\subsection{Assessing the excitation frequency and feedback design}
First, a numerical study is performed to find the frequency intervals in which the resonant substructure and the overall system are quasi-stable and stable, respectively. The results of this analysis are shown in \cref{fig:NCss_spectral_abscissa} in terms of spectral abscissas evaluated over the frequency range $\tilde\omega\in[2, 12]\:\mathrm{Hz}$. 
The frequency range is covered by a dense grid. For each frequency grid point, the feedback parameters are evaluated by \eqref{eq:control_design:g_2_tau_2}, which provides smaller values of the delay in this particular application compared to \eqref{eq:control_design:g_1_tau_1}, considering the delay branches $k=0$ and $k=1$. Then, the corresponding spectral abscissas \eqref{eq:alphaOS} and \eqref{eq:alphaRS} are obtained by applying the function \emph{tds\_sa} of TDS-CONTROL toolbox \cite{appeltans_2022_tdsCtrl}. 
\begin{figure}[t]
    \centering
    \includegraphics[scale=1]{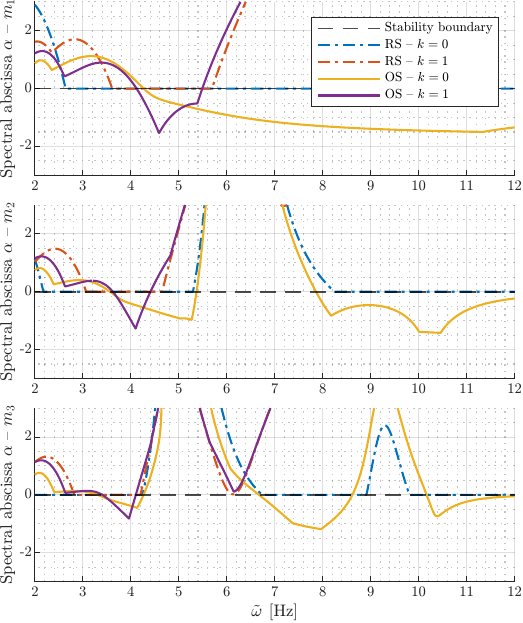}
    \caption{
    Spectral abscissas of the resonant substructure given by \eqref{eq:alphaRS} (RS -- dash-dot line) and of the overall setup given by \eqref{eq:alphaOS}  (OS -- solid line) for branches $k=0$ and $k=1$. %To achieve experimental validation for collocated case as well as for non-collocated cases at the same frequency $\omega = 4.20 \: \mathrm{Hz}$, branch $k=1$ is selected to stop mass $m_1$.
    }
    \label{fig:NCss_spectral_abscissa}
\end{figure}
This analysis is performed for all the considered cases $m_n,\ n=1, 2, 3$, i.e. for stopping $m_1$ (upper subfigure of \cref{fig:NCss_spectral_abscissa}), $m_2$ (middle subfigure of \cref{fig:NCss_spectral_abscissa}) and $m_3$ (lower subfigure of \cref{fig:NCss_spectral_abscissa}). The analysis provides the following applicable ranges (in Hz) for which both \eqref{eq:alphaOSstab} and \eqref{eq:alphaRSstab} are fulfilled simultaneously:
\begin{itemize}
    \item $m_1$: $\omega\in[4.27, 12]$ for $k=0$ and $\omega\in[4.13, 5.48]$ for $k=1$,
    \item $m_2$: $\omega\in[3.57, 5.28] \cup [8.26, 12]$ for $k=0$ and $\omega\in[3.63, 4.40]$ for $k=1$,
    \item $m_3$: $\omega\in[3.31, 4.26] \cup [6.75, 8.61] \cup [10.17, 12]$ for $k=0$ and $\omega\in[3.41, 4.10]$ for $k=1$.
\end{itemize}

As can be seen, the applicable ranges are relatively narrow already for taking the particular cases separately. Expectedly, the operating range becomes narrower when we search those frequencies at which all three masses can be stopped, one at a time,
\begin{equation}\label{eq:frequency-low}
    \tilde\omega\in[4.13, 4.22] 
\end{equation}
considering $k=1$ branch for $m_1$ and $k=0$ for $m_2$ and $m_3$,
and 
\begin{equation}\label{eq:frequency-high}
    \tilde\omega\in[8.2, 8.6]\cup[10.12, 12] 
\end{equation}
considering $k=0$ branch for all the targets.

\subsection{Analysis and experimental validation for low-frequency excitation}
In the analysis performed in \cite{vyhlidal_2019_analysis}, it was demonstrated that for the collocated case, the best performance of the DR in vibration suppression is achieved close to the resonant frequency of the passive absorber, which projects to the magnitude drop at the response of the overall system. Analogous results can be expected also for the non-collocated cases. As can be seen from the magnitude frequency responses of $P(\jmath \omega;0,0)$ by \eqref{eq:TF} in \cref{fig:case_study:spectral_sensitivity}, such minima appear for $\tilde\omega=\SI{4.42}{Hz}$ ($m_1$), $\tilde\omega=\SI{3.83}{Hz}$ ($m_2$) and $\tilde\omega=\SI{3.6}{Hz}$ ($m_3$). Balancing all these aspects, we select the excitation frequency $\omega=4.20 \: \mathrm{Hz}$ towards the experimental validation. Applying \eqref{eq:control_design:g_2_tau_2}, the following feedback parameters result:
\begin{itemize}
    \item stopping the body $m_1$ (collocated), with $k=1$
    \begin{equation}\label{eq:set1}
    \begin{aligned}
        g_1    &= -65.34 \: \mathrm{kg \, s^{-2}},
        \tau_1 &= 0.3263 \: \mathrm{s},
    \end{aligned}
\end{equation}
\item stopping the body $m_2$ (non-collocated), with $k=0$
\begin{equation}\label{eq:set2}
    \begin{aligned}
        g_2    &= -124.14 \: \mathrm{kg \, s^{-2}},
        \tau_2 &= 0.0165 \: \mathrm{s},
    \end{aligned}
\end{equation}
\item stopping the body $m_3$ (non-collocated), with $k=0$
\begin{equation}\label{eq:set3}
    \begin{aligned}
        g_3    &= -302.47 \: \mathrm{kg \, s^{-2}},
        \tau_3 &= 0.0146 \: \mathrm{s}.
    \end{aligned}
\end{equation}
\end{itemize}
As can be seen in the middle subfigure of \cref{fig:case_study:spectral_sensitivity} with amplitude responses, all the three controllers fully suppress vibration at the selected bodies at the frequency $\omega = 4.20 \: \mathrm{Hz}$, as required. However, notice that due to the {\em v-shape} of the characteristics close to the target frequency point, the robustness against the mismatch between the true and nominal frequencies is relatively small. This result also indicates the necessity of having very precise model of the resonant substructure for the feedback design.

\begin{figure}[t]
    \centering
    \includegraphics[scale=1]{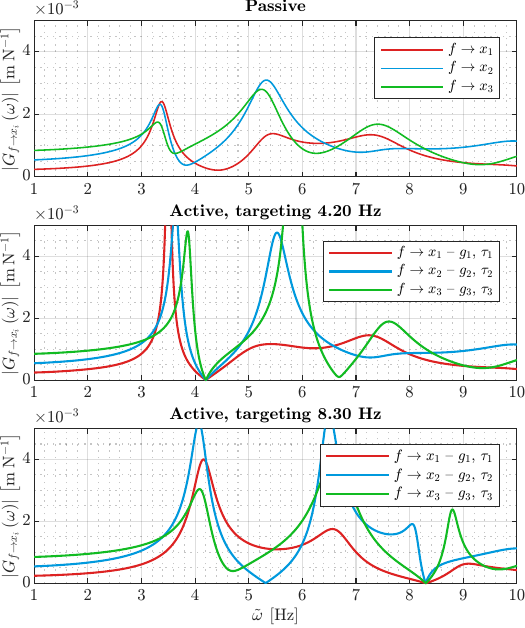}
    \caption{Amplitude frequency response of $P(\jmath \omega)$ by \eqref{eq:TF}. Without any control in the top, with controllers to stop individual carts at frequency $\omega=4.20 \: \mathrm{Hz}$ in the middle and at frequency $\omega=8.30 \: \mathrm{Hz}$ in the bottom.}
    \label{fig:case_study:spectral_sensitivity}
\end{figure}

% Note, however, that there is no guarantee that all three bodies can be silenced at the same frequency unless an additional controller is applied to stabilize the overall setup. 
 %It should also be noted that there is a common range at higher frequencies where all three bodies can be theoretically silenced, but we experience instability in this range. We suspect that custom ICUs are causing this instability due to a signal phase shift of the internal current-loop control. Unfortunately, these ICUs are black boxes and we cannot change the parameters of the internal current control loop.

%Having chosen the design frequency $\omega=4.20 \: \mathrm{Hz}$ and following \eqref{eq:control_design:g_tau_jomega}--\eqref{eq:control_design:g_2_tau_2}, the parameters of all three controllers were calculated. More specifically, to stop the body $m_1$ (the collocated case) the branch $k=1$ is used, which results in the parameter couple

To demonstrate the vibration suppression at the three different target bodies $m_1, m_2$ and $m_3$ for a single frequency $\omega=4.20 \: \mathrm{Hz}$, the following scenario is considered, with the experimental results shown in \cref{fig:exp_disp_all}, see also the video from the experiment\footnote{A video of the experiment is available at\\ \href{https://control.fs.cvut.cz/en/aclab/ncva}{https://control.fs.cvut.cz/en/aclab/ncva}}. The disturbance force $f$ given by \eqref{eq:general:excitation}, with $F=\SI{3}{N}$, starts to act on the mass $m_3$ at time $t = 5 \: \mathrm{s}$. After 10 seconds when we can observe the insufficient effect of passive absorption, the DR feedback \eqref{eq:general:cDR_position} with parameters \eqref{eq:set1} tuned to stop the mass $m_1$ is activated. After a short transient, the body $m_1$ is almost fully stopped. As can be seen in the detailed \cref{fig:exp_disp_m1}, the residual motion of measured $x_1$ is at the level of measurement (quantization) noise of position incremental sensor. The DR feedback is deactivated after \SI{15}{s}, i.e., at $t=30\: \mathrm{s}$. The passive regime lasts until $t=40\: \mathrm{s}$ when the DR feedback \eqref{eq:general:cDR_position} with parameters \eqref{eq:set2} tuned to stop the mass $m_2$ is activated. Similarly to the previous case, after a short transient, the body $m_2$ is almost fully stopped with residual swings of $x_3$ being at the measurement noise level as seen in the detailed \cref{fig:exp_disp_m2}. At $t=55\: \mathrm{s}$ the DR feedback is deactivated and the passive regime lasts till $t=65\: \mathrm{s}$, when the DR feedback \eqref{eq:general:cDR_position} with parameters \eqref{eq:set3} tuned to stop the mass $m_3$ is activated. The results are as good as for the previous two cases, despite the fact that the DR control action needs to propagate from $m_a$ through $m_1$, $m_2$ and all the flexible connections before it compensates the effect of the excitation force on $m_3$. Again, as seen in detailed \cref{fig:exp_disp_m3}, the residual deflections of $x_3$ is at the level of measurement noise. The DR feedback is turned off at $t=80\: \mathrm{s}$ and the experiment is completed by passive regime lasting till $t=86\: \mathrm{s}$. %According to our knowledge, it is for the first time when the spatial tuneability at the non-collocated vibration suppression problem was confirmed experimentally. 

\begin{figure}[t]
    \centering
    \includegraphics[scale=1]{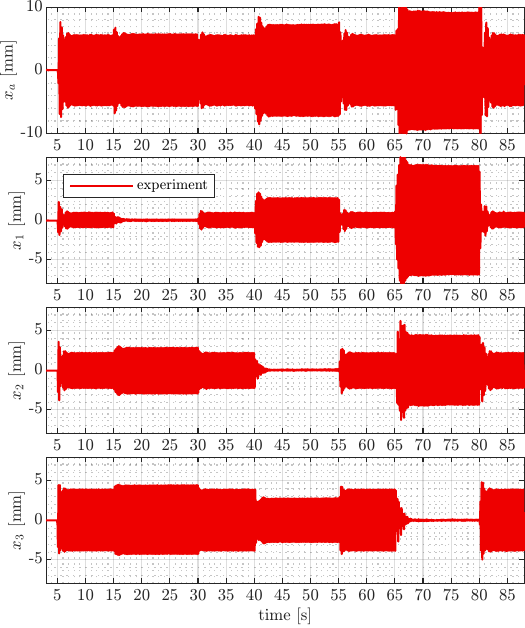}
    \caption{Experimental results of absorbing vibration excited by the disturbance force $f$ given by \eqref{eq:general:excitation}, with $F=\SI{3}{N}$ and $\omega=4.20 \: \mathrm{Hz}$, applying the DR position feedback \eqref{eq:general:cDR_position} with parameters: i) \eqref{eq:set1} active at $t\in[15, 30]\mathrm{s}$ to silence $x_1$ (collocated), ii) \eqref{eq:set2} active at $t\in[40, 55]\mathrm{s}$ to silence $x_2$ (non-collocated), and iii) \eqref{eq:set3} active at $t\in[65, 80]\mathrm{s}$ to silence $x_3$ (non-collocated).}
    \label{fig:exp_disp_all}
\end{figure}

%\subsection{Comparing the experimental and simulation results}
In a detailed look of the performance for $\omega=4.20\:\mathrm{Hz}$ in \cref{fig:exp_disp_m1}, \cref{fig:exp_disp_m2}, and \cref{fig:exp_disp_m3}, each DR setting is shown in comparison with simulations performed in Matlab-Simulink (ODE45 solver with \emph{RelTol} $10^{-6}$). A very good match between the simulation and the experimental results can be seen for each of the cases. 
% This is most likely the reason why the results in vibration suppression are so close to ideal. 
Notice that for all the three considered cases, the transients at silencing the target bodies are shorter for the experiments. It is due to the slip-stick effect of Coulomb friction, which naturally appears at the physical setup, but is not included in the linear model used for the simulations.
\begin{figure}
    \centering
    \includegraphics[scale=1]{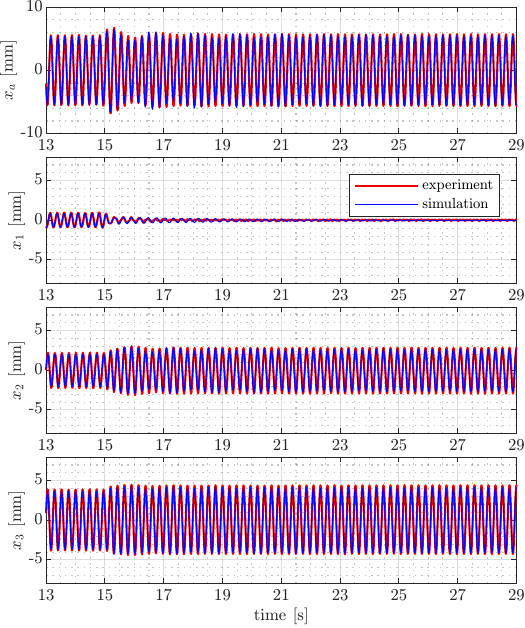} % 2 column layout
    \caption{Experimental results of collocated vibration absorption targeting $m_1$ -- detail of \cref{fig:exp_disp_all} in comparison with simulations.}
    \label{fig:exp_disp_m1}
\end{figure}
\begin{figure}
    \centering
    \includegraphics[scale=1]{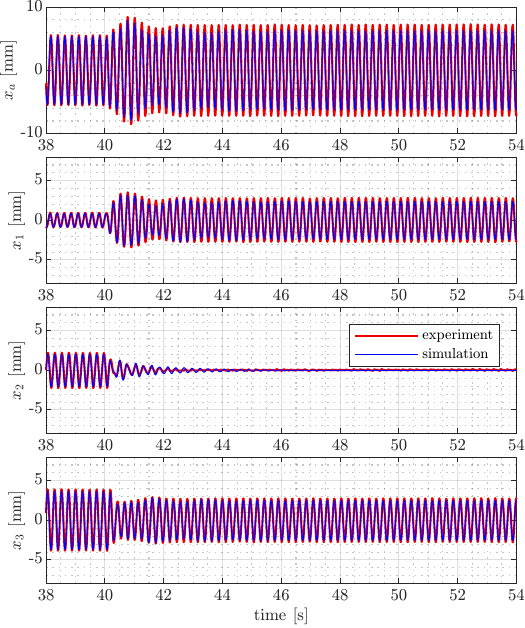} % double column
    \caption{Experimental results of non-collocated vibration absorption targeting $m_2$ -- detail of \cref{fig:exp_disp_all} in comparison with simulations.}
    \label{fig:exp_disp_m2}
\end{figure}
\begin{figure}
    \centering
    \includegraphics[scale=1]{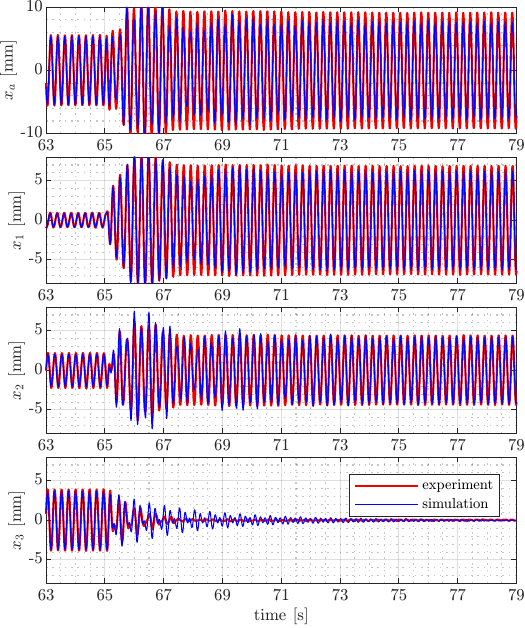} % double column
    \caption{Experimental results of non-collocated vibration absorption targeting $m_3$ -- detail of \cref{fig:exp_disp_all} in comparison with simulations.}
    \label{fig:exp_disp_m3}
\end{figure}

% I was personally surprised by the exceptional results of the non-collocated case where the body $m_3$ is stopped. First, because it works so well despite the long chain of bodies used to stop the body $m_3$ and despite the uncertainties in the parameters, especially in the damping. Second, it worked better than in the simulation. This might be caused by dry friction, which starts to play a significant role when the amplitude of the vibration becomes small, or by unmodeled nonlienarities.

\subsection{A note on higher frequency excitation}
To complete the analysis, we provide a short note on targeting excitation in the higher frequency range. From the admissible region \eqref{eq:frequency-high}, we select the excitation by $\omega=8.3 \: \mathrm{Hz}$. Applying \eqref{eq:control_design:g_2_tau_2}, with $k=0$ in all the cases, we obtain:
\begin{itemize}
    \item stopping the body $m_1$ (collocated), with $k=0$
    \begin{equation}\label{eq:set1h}
    \begin{aligned}
        g_1    &= -1011.59 \: \mathrm{N \, m^{-1}}, \tau_1 &= 0.0018 \: \mathrm{s},
    \end{aligned}
\end{equation}
\item stopping the body $m_2$ (non-collocated), with $k=0$
\begin{equation}\label{eq:set2h}
    \begin{aligned}
        g_2    &= -688.13 \: \mathrm{kg \, s^{-2}},
        \tau_2 &= 0.0073 \: \mathrm{s},
    \end{aligned}
\end{equation}
\item stopping the body $m_3$ (non-collocated), with $k=0$
\begin{equation}\label{eq:set3h}
    \begin{aligned}
        g_3    &=-956.08 \: \mathrm{kg \, s^{-2}},
        \tau_3 &= 0.0040 \: \mathrm{s},
    \end{aligned}
\end{equation}
\end{itemize}

%$f = 8.30 \: \mathrm{Hz}$:\\
% Stop m0 -- branch 1
%$g_1    = -1011.5922 \: \mathrm{N \, m^{-1}}$, 
%$\tau_1 = 0.0018 \: \mathrm{s}$,\\
% Stop m1 -- branch 0
%$g_2    = -688.1256 \: \mathrm{N \, m^{-1}}$, 
%$\tau_2 = 0.0073 \: \mathrm{s}$,\\
% Stop m2 -- branch 0
%$g_3    = -956.0772 \: \mathrm{N \, m^{-1}}$, 
%$\tau_3 = 0.0040 \: \mathrm{s}$
\begin{figure}[t]
    \centering
    \includegraphics[scale=1]{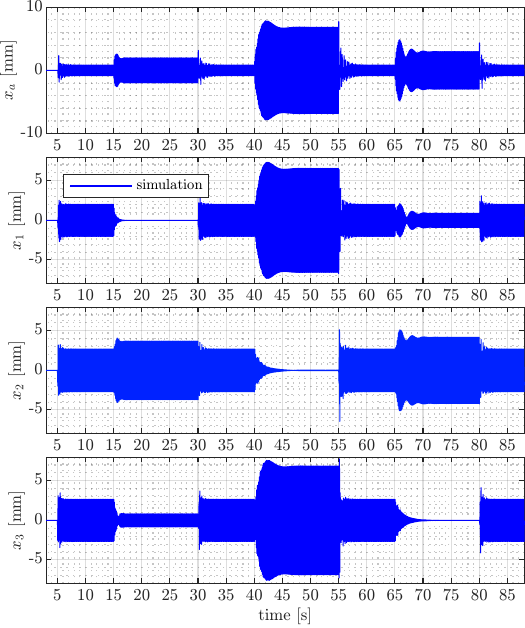} % double column
    \caption{Simulation results of absorbing vibration excited by the disturbance force $f$ given by \eqref{eq:general:excitation}, with $F=\SI{3}{N}$ and $\omega=8.30 \: \mathrm{Hz}$, applying the DR position feedback \eqref{eq:general:cDR_position} with parameters: i) \eqref{eq:set1h} active at $t\in[15, 30]\mathrm{s}$ to silence $x_1$ (collocated), ii) \eqref{eq:set2h} active at $t\in[40, 55]\mathrm{s}$ to silence $x_2$ (non-collocated), and iii) \eqref{eq:set3h} active at $t\in[65, 80]\mathrm{s}$ to silence $x_3$ (non-collocated).}
    \label{fig:case_study:sim_disp_all_8.30_Hz}
\end{figure}
The correctness of the synthesis is confirmed by both frequency domain analysis shown in bottom subfigure of \cref{fig:case_study:spectral_sensitivity}, where the amplitude is forced to zero in point-wise manner, and in \cref{fig:case_study:sim_disp_all_8.30_Hz}, where the simulation results are shown for the same scenario as in \cref{fig:exp_disp_all}. Unfortunately, the experimental validation cannot be performed on the current setup due to both actuation and hardware limitations. Notice that the gains in \eqref{eq:set1h}, \eqref{eq:set2h} and \eqref{eq:set3h} are considerably higher than in \eqref{eq:set1}, \eqref{eq:set2} and \eqref{eq:set3}, which naturally brings higher sensitivities to system-model mismatch and enhanced role of system non-linearities. Though, the main limitation stems in that the obtained values of the delays are too close to the sampling period $\Delta t=0.001 \: \mathrm{s}$. Remedy for this is to move to system with higher sampling speeds.

%It can also be seen that part of the original structure is used to perform the non-collocated vibration suppression, which may introduce a risk of structural fatigue. In future research, structural analysis and optimization will be performed to minimize the risk of fatigue.

\section{Concluding remarks}
\label{sec:conclusion}
Non-collocated vibration absorption using delayed resonator (with position feedback) as the tuning procedure is analyzed and experimentally validated. An easily scalable analytical pathway is presented to handle systems with higher degrees of freedom. The novelty comes in the formation of the absorber tuning feedback law. It can be obtained without the need to form the transfer function between the excitation force and target mass position. Detailed numerical and experimental validation is performed on a setup with three masses, for both collocated and non-collocated deployment. The spectral analysis revealed that even for the adopted setup with three masses only, it was difficult to find an excitation frequency for which the three masses can be sequentially silenced due to stability constraints. Although, with a proper system-model match and carefully tuned hardware, almost ideal vibration absorption was achieved not only for the collocated, but also for the non-collocated cases. To our best knowledge, it is for the first time such spatial tunability in vibration absorption is confirmed experimentally. Further research directions include synthesis of more complex control schemes to extend the applicable range of frequencies for which the system is stable, and enhancement of robustness against uncertainties in both system parameters and excitation frequencies. Additionally, hardware with higher sampling rate will be necessary to experimentally validate the intended results at higher frequency ranges of excitation.

% \clearpage % To be deleted in a final version
\bibliographystyle{elsarticle-num} 
\bibliography{NCss_Olgac,references}

\end{document}